\documentclass[11pt,a4paper]{article}
\usepackage{indentfirst,mathrsfs}
\usepackage{amsfonts,amsmath,amssymb,amsthm}
\usepackage{latexsym,amscd}
\usepackage{amsbsy}
\usepackage[all]{xy}
\usepackage{graphicx}
\usepackage{threeparttable}
\usepackage[usenames,dvipsnames,svgnames,table]{xcolor}
\usepackage{cite}

\newtheorem{thm}{Theorem}
\newtheorem{lem}{Lemma}

\newtheorem{example}{Example}
\newtheorem{remark}{Remark}

\begin{document}
\date{}
\title{MDS Symbol-Pair Codes from Repeated-Root Cyclic Codes}
\maketitle
\begin{center}
\author{\large Junru Ma \quad\quad Jinquan Luo
\footnote{Corresponding author
\par
The authors are with School of Mathematics
and Statistics \& Hubei Key Laboratory of Mathematical Sciences, Central China Normal University, Wuhan China, 430079.
\par
E-mails: junruma@mails.ccnu.edu.cn(J.Ma), luojinquan@mail.ccnu.edu.cn(J.Luo).
}}
\end{center}

\begin{quote}
  {\small {\bf Abstract:} \ \
  Symbol-pair codes are proposed to combat pair-errors in symbol-pair read channels.
  The minimum symbol-pair distance is of significance in determining the error-correcting capability of a symbol-pair code.
  Maximum distance separable (MDS) symbol-pair codes are optimal in the sense that such codes can achieve the Singleton bound.
  In this paper, two new classes of MDS symbol-pair codes are proposed utilizing repeated-root cyclic codes over finite fields with odd characteristic.
  Precisely, these codes poss minimum symbol-pair distance ten or twelve, which is bigger than all the known MDS symbol-pair codes from constacyclic codes.}

\end{quote}

\section{Introduction}

In information theory, noisy channels are analyzed generally by dividing the message into independent information units.
With the development of modern high-density data storage systems, the reading process may be lower than that of the process used to store the data.
Motivated by this situation, a new coding framework named symbol-pair code was proposed by Cassuto and Blaum (2010) to guard against pair-errors over symbol-pair channels in \cite{CB1}.
Cassuto and Blaum firstly studied symbol-pair codes on pair-error correctability conditions, code construction, decoding methods and asymptopic bounds in \cite{CB1,CB2}.
Shortly afterwards, Cassuto and Litsyn \cite{CL} established that codes for correcting pair-errors exist with strictly higher rates compared to codes for the Hamming metric with the same relative distance.
Later, researchers further investigated symbol-pair codes, including the construction of symbol-pair codes \cite{CJKWY,CKW,CLL,DGZZZ,DNS,EGY,KZL,KZZLC,LG,ML1}, some decoding algorithms of symbol-pair codes \cite{HMH,LXY,MHT,THM,YBS} and the symbol-pair weight distribution of some linear codes \cite{DNSS1,DNSS2,DWLS,ML2,SZW}.

The minimum symbol-pair distance plays an important role in determining the error-correcting capability of a symbol-pair code.
Cassuto and Blaum \cite{CB1} determined that a code $\mathcal{C}$ with minimum symbol-pair distance $d_{p}$ can correct up to $\lfloor\frac{d_{p}-1}{2}\rfloor$ symbol-pair errors.
In 2012, Chee et al. \cite{CKW} derived a Singleton-type bound on symbol-pair codes.
Similar to classical error-correcting codes, the symbol-pair codes achieving the Singleton-type bound are called MDS symbol-pair codes.
Recently, the construction of MDS symbol-pair codes has attracted the attention of many researchers.
In general, there are two methods to construct MDS symbol-pair codes.
The first one is based on linear codes with certain properties, such as MDS codes \cite{CJKWY,CKW} and constacyclic (cyclic) codes \cite{CLL,KZL,KZZLC,LG,ML1}.
The second method is to construct MDS symbol-pair codes by utilizing interleaving techniques \cite{CJKWY,CKW}, Eulerian graphs \cite{CJKWY,CKW}, projective geometry \cite{DGZZZ} and algebraic geometry codes over elliptic curves \cite{DGZZZ}.

In Table \ref{Tab-constacyclic-MDS}, we summarize all currently known MDS symbol-pair codes from constacyclic codes.
\begin{table}[!htb]
  \caption{Known MDS symbol-pair codes from constacyclic codes}\label{tab:tablenotes}
  \centering
  \begin{threeparttable}\label{Tab-constacyclic-MDS}
  \begin{tabular}{ccc}
  \hline $(n,\,d_{p})_{q}$      &Condition                &Reference\\
  \hline $\left(n,\,5\right)_{q}$       &$n\,|\left(q^{2}+q+1\right)$
    &\cite{KZL},\cite{LG} \\
  \hline $\left(n,\,6\right)_{q}$       &$n\,|\left(q^{2}+1\right)$
    &\cite{KZL},\cite{LG} \\
  \hline $\left(n,\,6\right)_{q}$
    &$n\,|\left(q^{2}-1\right)$,\,$n$\,odd or $n$ even and $v_{2}\left(n\right)<v_{2}\left(q^{2}-1\right)$     &\cite{LG} \\
  \hline $\left(n,\,6\right)_{q}$
    &$q\geq3,\,n\geq q+4,\,n\,|\left(q^{2}-1\right)$    &\cite{CLL} \\
  \hline $\left(lp,\,5\right)_{p}$
    &$p\geq5,\,l>2,\,{\rm gcd}\left(l,\,p\right)=1,\,l\,|\left(p-1\right)$
    &\cite{CLL} \\
  \hline $\left(p^{2}+p,\,6\right)_{p}$    &$p\geq3$          &\cite{KZZLC} \\
  \hline $\left(2p^{2}-2p,\,6\right)_{p}$  &$p\geq3$          &\cite{KZZLC} \\
  \hline $\left(3p,\,6\right)_{p}$         &$p\geq5$          &\cite{CLL}   \\
  \hline $\left(3p,\,7\right)_{p}$         &$p\geq5$          &\cite{CLL}   \\
  \hline $\left(4p,\,7\right)_{p}$   &$p\equiv 3\left({\rm mod}\,4\right)$ &\cite{KZZLC}\\
  \hline $\left(3p,\,8\right)_{p}$   &$3\,|\left(p-1\right)$  &\cite{CLL} \\
  \hline $\left(3p,\,10\right)_{p}$  &$3\,|\left(p-1\right)$
    &Theorem \ref{thmMDS3p10p} \\
  \hline $\left(3p,\,12\right)_{p}$        &$3\,|\left(p-1\right)$
    &Theorem \ref{thmMDS3p12p} \\
  \hline
  \end{tabular}
  \begin{tablenotes}
        \footnotesize
          \item  where $q$ is a power of prime $p$.
  \end{tablenotes}
  \end{threeparttable}
\end{table}
As we can see, most known codes in Table \ref{Tab-constacyclic-MDS} poss a fairly small symbol-pair distance.
The construction of symbol-pair codes with comparatively large minimum symbol-pair distance is a very interesting problem.
It is shown in\cite{DGZZZ} that there exist $q$-ary MDS symbol-pair codes from algebraic geometry codes over elliptic curves with larger minimum symbol-pair distance. But their lengths are bounded by $q+2\sqrt{q}$.
Inspired by the aforementioned works, in this paper, we propose two new classes of $p$-ary MDS symbol-pair codes with length $n=3p$ by employing repeated-root cyclic codes. Notably, these codes poss minimum symbol-pair distance $10$ or $12$, which is bigger than all the known MDS symbol-pair codes from constacyclic codes.

The rest of this paper is organized as follows. In Section $2$, we introduce some basic notations and results on symbol-pair codes and constacyclic codes.
By means of repeated-root cyclic codes, we investigate MDS symbol-pair codes in Section $3$.
In Section $4$, we make some conclusions.

\section{Preliminaries}

In this section, we review some basic notations and results on symbol-pair codes and constacyclic codes, which will be used to prove our main results in the sequel.

\subsection{Symbol-pair Codes}

Let $q=p^{m}$ and $\mathbb{F}_{q}$ denote the finite field with $q$ elements, where $p$ is a prime and $m$ is a positive integer.
Throughout this paper, let $\star$ be an element in $\mathbb{F}_{q}^{*}$ and $\mathbf{0}$ denotes the all-zero vector.
Let $n$ be a positive integer.
From now on, we always take the subscripts modulo $n$.
For any vector $\mathbf{x}=\left(x_{0}, x_{1}, \cdots, x_{n-1}\right)$ in $\mathbb{F}_{q}^{n}$, the symbol-pair read vector of $\mathbf{x}$ is
\begin{equation*}
  \pi\left(\mathbf{x}\right)=\left(\left(x_{0},\,x_{1}\right),
  \,\left(x_{1},\,x_{2}\right),\cdots,\left(x_{n-2},\,x_{n-1}\right),
  \,\left(x_{n-1},\,x_{0}\right)\right).
\end{equation*}
Observe that every vector $\mathbf{x}\in \mathbb{F}_{q}^{n}$ has a unique pair representation $\pi(\mathbf{x})$.
Denote by $\mathbb{Z}_{n}$ the residue class ring $\mathbb{Z}/n\mathbb{Z}$.
Recall that the {\it Hamming weigh} of $\mathbf{x}$ is
\begin{equation*}
  w_{H}(\mathbf{x})=\left|\left\{i\in \mathbb{Z}_{n}\,\big|\,x_{i}\neq 0\right\}\right|.
\end{equation*}
Accordingly, the {\it symbol-pair weight} of $\mathbf{x}$ is defined by
\begin{equation*}
  w_{p}\left(\mathbf{x}\right)=\left|\left\{i\in \mathbb{Z}_{n}\,\big|\,\left(x_{i},\,x_{i+1}\right)\neq \left(0,\,0\right)\right\}\right|.
\end{equation*}
For any two vectors $\mathbf{x},\,\mathbf{y}\in \mathbb{F}_{q}^{n}$, the {\it symbol-pair distance} between $\mathbf{x}$ and $\mathbf{y}$ is
\begin{equation*}
  d_{p}\left(\mathbf{x},\,\mathbf{y}\right)=\left|\left\{i\in \mathbb{Z}_{n}\,\big|\,\left(x_{i},\,x_{i+1}\right)\neq \left(y_{i},\,y_{i+1}\right)\right\}\right|.
\end{equation*}
A code $\mathcal{C}$ is said to have {\it minimum symbol-pair distance} $d_{p}$ if
\begin{equation*}
  d_{p}={\rm min}\left\{d_{p}\left(\mathbf{x},\,\mathbf{y}\right)\,\big|\, \mathbf{x},\,\mathbf{y}\in \mathcal{C},\mathbf{x}\neq\mathbf{y}\right\}.
\end{equation*}
Elements of $\mathcal{C}$ are called {\it codewords} in $\mathcal{C}$.
It is shown in \cite{CB1,CB2} that for any $0<d_H(\mathcal{C})<n$,
\begin{equation}\label{eqdhdp}
  d_H(\mathcal{C})+1\leq d_{p}\left(\mathcal{C}\right)\leq 2\cdot d_H(\mathcal{C}).
\end{equation}

Similar to classical error-correcting codes, the size of a symbol-pair code satisfies the following Singleton bound.

\begin{lem}{\rm(\!\cite{CJKWY}\,)}\label{Singleton}
Let $q\geq 2$ and $2\leq d_{p}\leq n$.
If $\mathcal{C}$ is a symbol-pair code with length $n$ and minimum symbol-pair distance $d_{p}$, then $\left|\mathcal{C}\right|\leq q^{n-d_{p}+2}$.
\end{lem}

The symbol-pair code achieving the Singleton bound is called a {\it maximum distance separable} (MDS) symbol-pair code.
For a linear code of length $n$, dimension $k$ and minimum symbol-pair distance $d_p$, if $d_p=n-k+1$, then it is called an {\it almost maximum distance separable} (AMDS) symbol-pair code.

\subsection{Constacyclic Codes}

In this subsection, we review some basic concepts of constacyclic codes.
For any $\eta\in\mathbb{F}_{q}^{*}$, the $\eta$-constacyclic shift $\tau_{\eta}$ on $\mathbb{F}_{q}^n$ is defined as
\begin{equation*}
  \tau_{\eta}\left(x_0,\,x_1,\cdots,x_{n-1}\right)
  =\left(\eta\,x_{n-1},\,x_0,\cdots,x_{n-2}\right).
\end{equation*}
A linear code $\mathcal{C}$ is an {\it$\eta$-constacyclic code} if $\tau_{\eta}\left(\mathbf{c}\right)\in\mathcal{C}$ for any codeword $\mathbf{c}\in\mathcal{C}$.
An $\eta$-constacyclic code is called a {\it cyclic code} if $\eta=1$ and a {\it negacyclic code} if $\eta=-1$.
Note that each codeword $\mathbf{c}=\left(c_0,\,c_1,\cdots,c_{n-1}\right)\in\mathcal{C}$ can be identified with a polynomial
\begin{equation*}
  c(x)=c_0+c_1\,x+\cdots+c_{n-1}\,x^{n-1}.
\end{equation*}
In this paper, we always regard the codeword $\mathbf{c}$ in $\mathcal{C}$ as the corresponding polynomial $c(x)$.
Indeed, a linear code $\mathcal{C}$ is an $\eta$-constacyclic code if and only if it is an ideal of the principle ideal ring $\mathbb{F}_{q}[x]/\langle x^n-\eta\rangle$.
Consequently, there is a unique monic polynomial $g(x)\in\mathbb{F}_{q}[x]$ with $g(x)\,|\,\left(x^n-\eta\right)$ and
\begin{equation*}
  \mathcal{C}=\left\langle g(x)\right\rangle=\left\{f(x)\,g(x)\,\left({\rm mod}\,x^n-\eta\right)\,\big|\,f(x)\in\mathbb{F}_{q}\left[x\right]\right\}.
\end{equation*}
We refer $g(x)$ as the {\it generator polynomial} of $\mathcal{C}$ and the dimension of $\mathcal{C}$ is $n-{\rm deg}\left(g(x)\right)$.

An $\eta$-constacyclic code of length $n$ over $\mathbb{F}_{q}$ is called a {\it simple-root} constacyclic code if $n$ and $p$ are relatively co-prime and a {\it repeated-root} constacyclic code if $p\,|\,n$.
Note that simple-root constacyclic codes can be characterized by their defining sets.
Furthermore, the BCH bound and the Hartmann-Tzeng bound for simple-root cyclic codes can be obtained by calculating the consecutive roots of the generator polynomial \cite{HP,MS}.
However, repeated-root cyclic codes cannot be directly characterized by sets of zeros.

Let $\mathcal{C}=\left\langle g(x)\right\rangle$ be a repeated-root cyclic code of length $lp^e$ over $\mathbb{F}_{q}$ with ${\rm gcd}\left(l,\,p\right)=1$ and
\begin{equation*}
  g(x)=\prod_{i=1}^{r}m_{i}(x)^{e_i}
\end{equation*}
the factorization of $g(x)$ into distinct monic irreducible polynomials $m_i(x)\in\mathbb{F}_q\left[x\right]$ of multiplicity $e_i$.
For any $0\leq t\leq p^e-1$, we denote $\overline{\mathcal{C}}_{t}$ the simple-root cyclic code of length $l$ over $\mathbb{F}_q$ with generator polynomial
\begin{equation*}
  \overline{g}_t(x)=\prod_{1\leq i\leq r,\,e_i>t}m_i(x).
\end{equation*}
If this product turns out to be $x^{l}-1$, then $\overline{\mathcal{C}}_{t}$ contains only the all-zero codeword and we set $d_H(\overline{\mathcal{C}}_{t})=\infty$.
If all $e_i(1\leq i\leq s)$ satisfy $e_i\leq t$, then we set $\overline{g}_t(x)=1$ and $d_H(\overline{\mathcal{C}}_{t})=1$.

The following lemma obtained from \cite{CMSS} indicates that the minimum Hamming distance of $\mathcal{C}$ can be derived from $d_{H}(\overline{\mathcal{C}}_{t})$, which will be used to determine the minimum Hamming distance of codes in Section $3$.

\begin{lem}{\rm(\!\cite{CMSS}\,)}\label{lemdistance}
Let $\mathcal{C}$ be a repeated-root cyclic code of length $lp^e$ over $\mathbb{F}_{q}$, where $l$ and $e$ are positive integers with ${\rm gcd}\left(l,\,p\right)=1$.
Then
\begin{equation}\label{eqdistance}
  d_{H}(\mathcal{C})={\rm min}\left\{P_{t}\cdot d_{H}\left(\overline{\mathcal{C}}_{t}\right)\,\big|\,0\leq t\leq p^e-1\right\}
\end{equation}
where
\begin{equation}\label{eqpt}
  P_{t}=w_{H}\left((x-1)^t\right)=\prod_{i}\left(t_{i}+1\right)
\end{equation}
with $t_{i}$'s being the coefficients of the radix-$p$ expansion of $t$.
\end{lem}

In the sequel, we recall the result of Lemma $3$ in \cite{ML1}, which will be used in Theorem \ref{thmMDS3p10p}.

\begin{lem}{\rm(\!\!\cite{ML1}\,)}\label{lemdH}
Let $\mathcal{C}$ be a repeated-root cyclic code of length $lp^{e}$ over $\mathbb{F}_{q}$ and $c(x)=\left(x^{l}-1\right)^{t}v(x)$ a codeword in $\mathcal{C}$ with Hamming weight $d_{H}(\mathcal{C})$, where $l$ and $e$ are positive integers with ${\rm gcd}\left(l,\,p\right)=1$, $0\leq t\leq p^{e}-1$ and $\left(x^{l}-1\right)\nmid v(x)$.
Then
\begin{equation*}
  w_{H}\left(c(x)\right)=P_{t}\cdot N_{v}
\end{equation*}
where $P_{t}$ is defined as $\left(\ref{eqpt}\right)$ in Lemma \ref{lemdistance} and $N_{v}=w_{H}\left(v(x)\,{\rm mod}\left(x^{l}-1\right)\right)$.
\end{lem}

In this paper, we will employ repeated-root cyclic codes to construct new MDS symbol-pair codes.
The following two lemmas will be applied in our later proof.

\begin{lem}{\rm(\!\cite{CLL})}\label{lemMDS}
Let $\mathcal{C}$ be an $[n,\,k,\,d_{H}]$ constacyclic code over $\mathbb{F}_{q}$ with $2\leq d_{H}<n$.
Then we have $d_{p}(\mathcal{C})\geq d_{H}+2$ if and only if $\mathcal{C}$ is not an MDS code, i.e., $k<n-d_{H}+1$.
\end{lem}

\begin{lem}{\rm(\!\cite{CLL})}\label{lemdh3}
Let $\mathcal{C}$ be an $\left[lp^e,\,k,\,d_{H}\right]$ repeated-root cyclic code over $\mathbb{F}_{q}$ and $g(x)$ the generator polynomial of $\mathcal{C}$, where ${\rm gcd}\left(l,p\right)=1$ and $l,\,e>1$.
If $d_{H}(\mathcal{C})$ is prime and one of the following two conditions is satisfied

(1) $l<d_{H}(\mathcal{C})<lp^e-k$;

(2) $x^l-1$ is a divisor of $g(x)$ and $2<d_{H}(\mathcal{C})<lp^e-k$,
\\[2mm]
then $d_{p}\left(\mathcal{C}\right)\geq d_{H}(\mathcal{C})+3$.
\end{lem}

\section{Constructions of MDS Symbol-Pair Codes}

In this section, for $n=3p$, we propose two new classes of MDS symbol-pair codes from repeated-root cyclic codes by analyzing the system of certain linear equations over $\mathbb{F}_{p}$.
Interestingly, the minimum symbol-pair distance of these codes ranges in $\{10,\,12\}$, which is bigger than all the known codes in Table \ref{Tab-constacyclic-MDS}.
For preparation, we define the following notations.

Let $n$ and $A_{i}$ be positive integers for any $1\leq i\leq n$ and
\begin{equation*}
  \mathcal{V}\left(A_{1},\cdots,A_{n}\right)=\left(A_{1}\,{\rm mod\,3},\cdots,A_{n}\,{\rm mod\,3}\right).
\end{equation*}
Denote by $\mathcal{CS}\left(A_{1},\cdots,A_{n}\right)=\left(a_{1},\cdots,a_{n}\right)$ the rearrangement of $\mathcal{V}\left(A_{1},\cdots,A_{n}\right)$ with $a_{i}\leq a_{j}$ for any $i<j$.
For instance, $\mathcal{CS}\left(5,10,4\right)=\left(1,1,2\right)$.

Now we present a class of MDS symbol-pair codes with length $3p$ and minimum symbol-pair distance $10$.

\begin{thm}\label{thmMDS3p10p}
Let $p$ be an odd prime with $3\,|\left(p-1\right)$.
Then there exists an MDS $\left(3p,\,10\right)_{p}$ symbol-pair code.
\end{thm}

\begin{proof}
Let $\mathcal{C}$ be a repeated-root cyclic code of length $3p$ over $\mathbb{F}_p$ with generator polynomial
\begin{equation*}
  g(x)=\left(x-1\right)^{4}\left(x-\omega\right)^2\left(x-\omega^2\right)^2
\end{equation*}
where $\omega$ is a primitive third root of unity in $\mathbb{F}_{p}$.

Note that Lemma \ref{lemdistance} yields that $\mathcal{C}$ is a $\left[3p,\,3p-8,\,5\right]$ cyclic code.
Precisely, recall that $\overline{g}_{t}(x)$ is the generator polynomial of $\overline{\mathcal{C}}_{t}$.
If $t\in \{0,\,1\}$, then $\overline{g}_{t}(x)=x^{3}-1$ and
\begin{equation*}
  P_{t}\cdot d_{H}\left(\overline{\mathcal{C}}_{t}\right)=\infty.
\end{equation*}
If $t=2$, then $\overline{g}_{2}(x)=x-1$ and
\begin{equation*}
  P_{2}\cdot d_{H}\left(\overline{\mathcal{C}}_{2}\right)=3\cdot 2=6.
\end{equation*}
If $t=3$, then $\overline{g}_{3}(x)=x-1$ and
\begin{equation*}
  P_{3}\cdot d_{H}\left(\overline{\mathcal{C}}_{3}\right)=4\cdot 2=8.
\end{equation*}
If $4\leq t\leq p-1$, then $\overline{g}_{t}(x)=1$ and
\begin{equation*}
  P_{t}\cdot d_{H}\left(\overline{\mathcal{C}}_{t}\right)=t+1\geq 5.
\end{equation*}
Due to the equality $\left(\ref{eqdistance}\right)$, one immediately has $d_{H}(\mathcal{C})=5$.

Since $\left(x^3-1\right)\,\big|\,g(x)$ and $2<5=d_{H}(\mathcal{C})<3p-(3p-8)=8$, by Lemma \ref{lemdh3}, one gets $d_{p}(\mathcal{C})\geq 8$.
Suppose that there exists a codeword $f(x)$ in $\mathcal{C}$ with Hamming weight $7$ such that $f(x)$ has $7$ consecutive nonzero entries.
Denote
\begin{equation*}
  f(x)=f_{0}+f_{1}\,x+f_{2}\,x^{2}+f_{3}\,x^{3}+f_{4}\,x^{4}
  +f_{5}\,x^{5}+f_{6}\,x^{6}
\end{equation*}
where $f_{i}\in \mathbb{F}_{p}^{*}$ for any $0\leq i\leq 6$.
It follows that the degree of $g(x)$ is less than or equal to the degree of $f(x)$, i.e., $8\leq 6$, which is impossible.
Hence there does not exist a codeword in $\mathcal{C}$ with Hamming weight $7$ and symbol-pair weight $8$.
Similarly, it can be verified that there does not exist a codeword in $\mathcal{C}$ with Hamming weight $8$ and symbol-pair weight $9$.

In the sequel, we claim that there does not exist a codeword in $\mathcal{C}$ with Hamming weight $5$ and symbol-pair weight $8$ (or $9$).
Let $c(x)$ be a codeword in $\mathcal{C}$ with Hamming weight $5$.
Assume that $c(x)$ has factorization $c(x)=\left(x^{3}-1\right)^{t}v(x)$,
where $0\leq t\leq p-1$, $\left(x^{3}-1\right)\nmid v(x)$ and $v(x)=v_{0}(x^{3})+x\,v_{1}(x^{3})+x^{2}\,v_{2}(x^{3})$.
Then by Lemma \ref{lemdH}, one can conclude that
\begin{equation*}
  5=w_{H}\left(\left(x^{3}-1\right)^{t}\right)\cdot w_{H}\left(v(x)\,{\rm mod}\left(x^{3}-1\right)\right)=\left(1+t\right)N_{v}
\end{equation*}
where $N_{v}=w_{H}\left(v(x)\,{\rm mod}\left(x^{3}-1\right)\right)$.
It follows that $\left(N_{v},\,t\right)=\left(1,\,4\right)$, which implies that the symbol-pair weight of $c(x)$ cannot be $8$ (or $9$).

In order to derive that $\mathcal{C}$ is an MDS $\left(3p,\,10\right)_{p}$ symbol-pair code, we need to prove that there does not exist a codeword $c(x)$ in $\mathcal{C}$ with $(w_{H}(c(x)),w_{p}(c(x)))=(6,\,8)$, $(6,\,9)$ or $(7,\,9)$.

Firstly, on the contrary, suppose that $c(x)$ is a codeword in $\mathcal{C}$ with Hamming weight $6$ and symbol-pair weight $8$.
Then its certain cyclic shift must have the form
\begin{equation*}
  \left(\star,\,\star,\,\star,\,\star,\,\star,\,\mathbf{0},\,\star,\,\mathbf{0}\right),
\end{equation*}
\begin{equation*}
  \left(\star,\,\star,\,\star,\,\star,\,\mathbf{0},\,\star,\,\star,\,\mathbf{0}\right)
\end{equation*}
or
\begin{equation*}
  \left(\star,\,\star,\,\star,\,\mathbf{0},\,\star,\,\star,\,\star,\,\mathbf{0}\right).
\end{equation*}
Without loss of generality, in this paper, we always suppose that the first coordinate of a codeword is $1$.
\begin{itemize}
  \item
  For the subcase of $\left(\star,\,\star,\,\star,\,\star,\,\star,\,\mathbf{0},\,\star,\,\mathbf{0}\right)$.
  Let $c(x)=1+a_{1}\,x+a_{2}\,x^{2}+a_{3}\,x^{3}+a_{4}\,x^{4}+a_{5}\,x^{l}$ with $6\leq l\leq 3p-2$ and $a_{i}\in \mathbb{F}_p^{*}$ for any $1\leq i\leq5$.
  \begin{itemize}
    \item
    If $\mathcal{V}\left(l\right)\in \{0,\,1\}$, then by $c\left(1\right)=c\left(\omega\right)=c\left(\omega^2\right)=0$, one can immediately obtain
    \begin{equation*}
    \left\{
    \begin{array}{l}
      1+a_{1}+a_{2}+a_{3}+a_{4}+a_{5}=0,\\[2mm]
      1+a_{1}\,\omega+a_{2}\,\omega^2+a_{3}+a_{4}\,\omega+a_{5}\,\omega^{l}=0,\\[2mm]
      1+a_{1}\,\omega^2+a_{2}\,\omega+a_{3}+a_{4}\,\omega^2+a_{5}\,\omega^{2l}=0.\\
    \end{array}
    \right.
    \end{equation*}
    This leads to $a_{2}=0$, a contradiction.

    \item
    If $\mathcal{V}\left(l\right)=2$, then $c^{\left(1\right)}\left(1\right)=c^{\left(1\right)}\left(\omega\right)=c^{\left(1\right)}\left(\omega^2\right)=0$ indicates $a_{3}=0$, a contradiction.
  \end{itemize}

  \item
  For the subcase of $\left(\star,\,\star,\,\star,\,\star,\,\mathbf{0},\,\star,\,\star,\,\mathbf{0}\right)$.
  Let $c(x)=1+a_{1}\,x+a_{2}\,x^{2}+a_{3}\,x^{3}+a_{4}\,x^{l}+a_{5}\,x^{l+1}$ with $5\leq l\leq 3p-3$ and $a_{i}\in \mathbb{F}_p^{*}$ for any $1\leq i\leq5$.
  \begin{itemize}
    \item
    If $\mathcal{V}\left(l\right)\in \{0,\,2\}$, then it follows from $c\left(1\right)=c\left(\omega\right)=c\left(\omega^2\right)=0$ that $a_{1}=0$ or $a_{2}=0$, a contradiction.

    \item
    If $\mathcal{V}\left(l\right)=1$, then by $c^{\left(1\right)}\left(1\right)=c^{\left(1\right)}\left(\omega\right)=c^{\left(1\right)}\left(\omega^2\right)=0$, one immediately has $a_{3}=0$, a contradiction.
  \end{itemize}

  \item
  For the subcase of $\left(\star,\,\star,\,\star,\,\mathbf{0},\,\star,\,\star,\,\star,\,\mathbf{0}\right)$.
  Let $c(x)=1+a_{1}\,x+a_{2}\,x^{2}+a_{3}\,x^{l}+a_{4}\,x^{l+1}+a_{5}\,x^{l+2}$ with $4\leq l\leq 3p-4$ and $a_{i}\in \mathbb{F}_p^{*}$ for any $1\leq i\leq5$.
  Then for any $4\leq l\leq 3p-4$, by $c^{\left(1\right)}\left(1\right)=c^{\left(1\right)}\left(\omega\right)=c^{\left(1\right)}\left(\omega^2\right)=0$, one can derive $a_{i}=0$ for some $3\leq i\leq5$, a contradiction.
\end{itemize}

Secondly, assume that there exists a codeword $c(x)$ in $\mathcal{C}$ with Hamming weight $6$ and symbol-pair weight $9$.
There are three subcases to be considered:
\begin{itemize}
  \item
  For the subcase of $c(x)=1+a_{1}\,x+a_{2}\,x^{2}+a_{3}\,x^{3}+a_{4}\,x^{l_{1}}+a_{5}\,x^{l_{2}}$ with $5\leq l_{1}<l_{2}\leq 3p-2$ and $a_{i}\in \mathbb{F}_p^{*}$ for any $1\leq i\leq5$.
  \begin{itemize}
    \item
    If $\mathcal{CS}\left(l_{1},\,l_{2}\right)=(1,\,2)$, then by $c^{\left(1\right)}\left(1\right)=c^{\left(1\right)}\left(\omega\right)=c^{\left(1\right)}\left(\omega^2\right)=0$, one can obtain that $a_{3}=0$, a contradiction.

    \item
    If $\mathcal{CS}\left(l_{1},\,l_{2}\right)\ne(1,\,2)$, then $c\left(1\right)=c\left(\omega\right)=c\left(\omega^2\right)=0$ indicates that $a_{1}=0$ or $a_{2}=0$, a contradiction.
  \end{itemize}

  \item
  For the subcase of $c(x)=1+a_{1}\,x+a_{2}\,x^{2}+a_{3}\,x^{l_{1}}+a_{4}\,x^{l_{1}+1}+a_{5}\,x^{l_{2}}$ with $4\leq l_{1}<l_{2}\leq 3p-2$ and $a_{i}\in \mathbb{F}_p^{*}$ for any $1\leq i\leq5$.
  \begin{itemize}
    \item
    If $\mathcal{V}\left(l_{1},\,l_{2}\right)\in\{(0,\,2),\,(1,\,0),\,(2,\,1)\}$, then by $c^{\left(1\right)}\left(1\right)=c^{\left(1\right)}\left(\omega\right)=c^{\left(1\right)}\left(\omega^2\right)=0$, one can immediately derive $a_{i}=0$ for some $3\leq i\leq5$, a contradiction.

    \item
    If $\mathcal{V}\left(l_{1},\,l_{2}\right)\notin\{(0,\,2),\,(1,\,0),\,(2,\,1)\}$, then $c\left(1\right)=c\left(\omega\right)=c\left(\omega^2\right)=0$ implies that $1=0$ or $a_{i}=0$ for some $i\in\{1,\,2\}$, a contradiction.
  \end{itemize}

  \item
  For the subcase of $c(x)=1+a_{1}\,x+a_{2}\,x^{l_{1}}+a_{3}\,x^{l_{1}+1}+a_{4}\,x^{l_{2}}+a_{5}\,x^{l_{2}+1}$ with $3\leq l_{1}<l_{2}\leq 3p-3$ and $a_{i}\in \mathbb{F}_p^{*}$ for any $1\leq i\leq5$.
  \begin{itemize}
    \item
    If $\mathcal{CS}\left(l_{1},\,l_{2}\right)=(0,\,0)$, then one can deduce that $p\nmid\, l_{1}\,l_{2}\left(l_{2}-l_{1}\right)$.
    By $c\left(1\right)=c\left(\omega\right)=c\left(\omega^2\right)=0$, one immediately gets
    \begin{equation*}
      1+a_{2}+a_{4}=a_{1}+a_{3}+a_{5}=0.
    \end{equation*}
    It follows from $c^{\left(1\right)}\left(1\right)=c^{\left(1\right)}\left(\omega\right)=c^{\left(1\right)}\left(\omega^2\right)=0$ that
    \begin{equation*}
      l_{1}\,a_{2}+l_{2}\,a_{4}=l_{1}\,a_{3}+l_{2}\,a_{5}=0.
    \end{equation*}
    Then one can conclude that $a_{2}+a_{3}=a_{4}+a_{5}=0$ due to $c^{\left(2\right)}\left(1\right)=0$.
    Hence the fact $c^{\left(3\right)}\left(1\right)=0$ indicates that $l_{2}\left(l_{2}-l_{1}\right)a_{5}=0$, a contradiction.

    \item
    If $\mathcal{CS}\left(l_{1},\,l_{2}\right)=(1,\,2)$, then by $c^{\left(1\right)}\left(1\right)=c^{\left(1\right)}\left(\omega\right)=c^{\left(1\right)}\left(\omega^2\right)=0$, one can immediately deduce $a_{3}=0$ or $a_{5}=0$, a contradiction.

    \item
    If $\mathcal{CS}\left(l_{1},\,l_{2}\right)\notin\{(0,\,0),\,(1,\,2)\}$, then $c\left(1\right)=c\left(\omega\right)=c\left(\omega^2\right)=0$ yields that $1=0$ or $a_{i}=0$ for some $1\leq i\leq5$, a contradiction.
  \end{itemize}
\end{itemize}

Thirdly, suppose that $c(x)$ is a codeword in $\mathcal{C}$ with Hamming weight $7$ and symbol-pair weight $9$.
Then we ought to discuss the following three subcases:
\begin{itemize}
  \item
  For the subscase of $c(x)=1+a_{1}\,x+a_{2}\,x^{2}+a_{3}\,x^{3}+a_{4}\,x^{4}+a_{5}\,x^{5}+a_{6}\,x^{l}$ with $7\leq l\leq 3p-2$ and $a_{i}\in \mathbb{F}_p^{*}$ for any $1\leq i\leq6$.
  \begin{itemize}
    \item
    If $\mathcal{V}\left(l\right)=0$, then by $c\left(1\right)=c\left(\omega\right)=c\left(\omega^2\right)=0$, one can immediately obtain $a_{1}+a_{4}=0$.
    It follows from $c^{\left(1\right)}\left(1\right)=c^{\left(1\right)}\left(\omega\right)=c^{\left(1\right)}\left(\omega^2\right)=0$ that $a_{1}+4\,a_{4}=0$.
    Hence $3\,a_{4}=0$, a contradiction.

    \item
    If $\mathcal{V}\left(l\right)\in \{1,\,2\}$, then $c^{\left(1\right)}\left(1\right)=c^{\left(1\right)}\left(\omega\right)=c^{\left(1\right)}\left(\omega^2\right)=0$ implies that $a_{3}=0$, a contradiction.
    \end{itemize}

  \item
  For the subcase of $c(x)=1+a_{1}\,x+a_{2}\,x^{2}+a_{3}\,x^{3}+a_{4}\,x^{4}+a_{5}\,x^{l}+a_{6}\,x^{l+1}$ with $6\leq l\leq 3p-3$ and $a_{i}\in \mathbb{F}_p^{*}$ for any $1\leq i\leq6$.
  \begin{itemize}
    \item
    If $\mathcal{V}\left(l\right)\in\{0,\,1\}$, then by $c^{\left(1\right)}\left(1\right)=c^{\left(1\right)}\left(\omega\right)=c^{\left(1\right)}\left(\omega^2\right)=0$, one has $a_{2}=0$ or $a_{3}=0$, a contradiction.

    \item
    If $\mathcal{V}\left(l\right)=2$, then $c\left(1\right)=c\left(\omega\right)=c\left(\omega^2\right)=c^{\left(1\right)}\left(1\right)=c^{\left(1\right)}\left(\omega\right)=c^{\left(1\right)}\left(\omega^2\right)=0$ indicates that $a_{4}=0$, a contradiction.
    \end{itemize}

  \item
  For the subcase of $c(x)=1+a_{1}\,x+a_{2}\,x^{2}+a_{3}\,x^{3}+a_{4}\,x^{l}+a_{5}\,x^{l+1}+a_{6}\,x^{l+2}$ with $5\leq l\leq 3p-4$ and $a_{i}\in \mathbb{F}_p^{*}$ for any $1\leq i\leq6$.
  For any $5\leq l\leq 3p-4$, $c\left(1\right)=c\left(\omega\right)=c\left(\omega^2\right)=c^{\left(1\right)}\left(1\right)
  =c^{\left(1\right)}\left(\omega\right)=c^{\left(1\right)}\left(\omega^2\right)=0$ yields that $a_{i}=0$ for some $4\leq i\leq6$, a contradiction.
\end{itemize}

Consequently, $\mathcal{C}$ is an MDS $\left(3p,\,10\right)_{p}$ symbol-pair code.
This completes the proof.
\end{proof}

In what follows, we construct a class of MDS symbol-pair codes with length $3p$ and minimum symbol-pair distance $12$, which is the maximum minimum symbol-pair distance for all known MDS symbol-pair codes from constacyclic codes.

\begin{thm}\label{thmMDS3p12p}
Let $p$ be an odd prime with $3\,|\left(p-1\right)$.
Then there exists an MDS $\left(3p,\,12\right)_{p}$ symbol-pair code.
\end{thm}

\begin{proof}
Let $\mathcal{C}$ be a repeated-root cyclic code of length $3p$ over $\mathbb{F}_p$ with generator polynomial
\begin{equation*}
  g(x)=\left(x-1\right)^{5}\left(x-\omega\right)^{3}\left(x-\omega^{2}\right)^{2}
\end{equation*}
where $\omega$ is a primitive third root of unity in $\mathbb{F}_{p}$.
It follows from Lemma \ref{lemdistance} that $\mathcal{C}$ is a $[3p,\,3p-10,\,6]$ code.
Since $\mathcal{C}$ is not MDS, by Lemma \ref{lemMDS}, one can get $d_{p}(\mathcal{C})\geq 8$.
With a similar manner to the proof of Theorem \ref{thmMDS3p10p}, it can be verified that there does not exist a codeword $c(x)$ in $\mathcal{C}$ with $(w_{H}(c(x)),w_{p}(c(x)))=(9,\,10)$ or $(10,\,11)$.
Besides, the proof of Theorem \ref{thmMDS3p10p} yields that there does not exist a codeword $c(x)$ in $\mathcal{C}$ with $(w_{H}(c(x)),w_{p}(c(x)))=(6,\,8)$, $(6,\,9)$, $(7,\,8)$, $(7,\,9)$ or $(8,\,9)$.

To determine that $\mathcal{C}$ is an MDS $\left(3p,\,12\right)_{p}$ symbol-pair code, it is sufficient to derive that there does not exist a codeword $c(x)$ in $\mathcal{C}$ with $(w_{H}(c(x)),w_{p}(c(x)))=(6,\,10)$, (6,\,11)$, (7,\,10)$, $(7,\,11)$, $(8,\,10)$, $(8,\,11)$ or $(9,\,11)$.

Firstly, we suppose that $c(x)$ is a codeword in $\mathcal{C}$ with Hamming weight $6$ and symbol-pair weight $10$.
Without loss of generality, we just consider the following two subcases:
\begin{itemize}
  \item
  For the subcase of $c(x)=1+a_{1}\,x+a_{2}\,x^{2}+a_{3}\,x^{l_{1}}+a_{4}\,x^{l_{2}}+a_{5}\,x^{l_{3}}$ with $4\leq l_{1}<l_{2}<l_{3}\leq 3p-2$ and $a_{i}\in \mathbb{F}_p^{*}$ for any $1\leq i\leq5$.
  \begin{itemize}
    \item
    If $\mathcal{CS}\left(l_{1},\,l_{2},\,l_{3}\right)=(0,\,1,\,2)$, then it can be checked that by $c\left(1\right)=c\left(\omega\right)=c\left(\omega^2\right)=0$, one can immediately get $1=0$ or $a_{i}=0$ for some $1\leq i\leq 5$, a contradiction.

    \item
    If $\mathcal{CS}\left(l_{1},\,l_{2},\,l_{3}\right)\ne(0,\,1,\,2)$, then $c\left(1\right)=c\left(\omega\right)=c\left(\omega^2\right)=c^{\left(1\right)}\left(1\right)
    =c^{\left(1\right)}\left(\omega\right)=c^{\left(1\right)}\left(\omega^2\right)=0$ yields that $a_{i}=0$ for some $1\leq i\leq5$, a contradiction.
  \end{itemize}

  \item
  For the subcase of $c(x)=1+a_{1}\,x+a_{2}\,x^{l_{1}}+a_{3}\,x^{l_{1}+1}+a_{4}\,x^{l_{2}}+a_{5}\,x^{l_{3}}$ with $3\leq l_{1}<l_{2}<l_{3}\leq 3p-2$ and $a_{i}\in \mathbb{F}_p^{*}$ for any $1\leq i\leq5$.
  \begin{itemize}
    \item
    If $\mathcal{V}\left(l_{1},\,l_{2},\,l_{3}\right)\in\left\{(0,\,0,\,0),\,(0,\,1,\,1),\,(0,\,2,\,2)
    ,\,(1,\,0,\,2),\,(1,\,2,\,0),\,(2,\,1,\,2),\,(2,\,2,\,1)\right\}$, then by $c\left(1\right)=c\left(\omega\right)=c\left(\omega^2\right)=c^{\left(1\right)}\left(1\right)=c^{\left(1\right)}\left(\omega\right)=c^{\left(1\right)}\left(\omega^2\right)=0$, it follows that $a_{i}=0$ for some $1\leq i\leq 5$, a contradiction.

    \item
    If $\mathcal{V}\left(l_{1},\,l_{2},\,l_{3}\right)\in\left\{(0,\,1,\,0),\,(0,\,0,\,1)\right\}$, then by $c^{\left(1\right)}\left(1\right)=c^{\left(1\right)}\left(\omega\right)=c^{\left(1\right)}\left(\omega^2\right)=c^{\left(2\right)}\left(1\right)=c^{\left(2\right)}\left(\omega\right)=0$, one can immediately obtain $a_{4}=0$ or $a_{5}=0$ since     $p\nmid\,l_{3}\left(l_{3}-l_{1}\right)$ or $p\nmid\,l_{2}\left(l_{2}-l_{1}\right)$.
    This leads to a contradiction.

    \item
    For other conditions, it can be verified that $c\left(1\right)=c\left(\omega\right)=c\left(\omega^2\right)=0$ indicates  $1=0$ or $a_{i}=0$ for some $1\leq i\leq 5$, which is impossible.
  \end{itemize}
\end{itemize}

Secondly, we assume that $c(x)$ is a codeword in $\mathcal{C}$ with Hamming weight $6$ and symbol-pair weight $11$.
Let $c(x)=1+a_{1}\,x+a_{2}\,x^{l_{1}}+a_{3}\,x^{l_{2}}+a_{4}\,x^{l_{3}}+a_{5}\,x^{l_{4}}$ with $3\leq l_{1}<l_{2}<l_{3}<l_{4}\leq 3p-2$ and $a_{i}\in \mathbb{F}_p^{*}$ for any $1\leq i\leq5$.
\begin{itemize}
  \item
  If $\mathcal{CS}\left(l_{1},\,l_{2},\,l_{3},\,l_{4}\right)\in\left\{(0,\,0,\,0,\,1),\,
  (0,\,1,\,1,\,1),\,(0,\,1,\,2,\,2)\right\}$, then by $c\left(1\right)=c\left(\omega\right)=c\left(\omega^2\right)=c^{\left(1\right)}\left(1\right)=c^{\left(1\right)}\left(\omega\right)=c^{\left(1\right)}\left(\omega^2\right)=0$ indicates that $a_{i}=0$ for some $2\leq i\leq 5$, a contradiction.

  \item
  If $\mathcal{CS}\left(l_{1},\,l_{2},\,l_{3},\,l_{4}\right)=(0,\,0,\,1,\,1)$, then by $c^{\left(1\right)}\left(1\right)=c^{\left(1\right)}\left(\omega\right)=c^{\left(1\right)}\left(\omega^2\right)=c^{\left(2\right)}\left(1\right)=c^{\left(2\right)}\left(\omega\right)=0$, one has $a_{i}=0$ for some $2\leq i\leq 5$, a contradiction.

  \item
  If $\mathcal{CS}\left(l_{1},\,l_{2},\,l_{3},\,l_{4}\right)\notin\{(0,\,1,\,1,\,1),\,(0,\,1,\,2,\,2),\,(0,\,0,\,1,\,1),\,(0,\,0,\,0,\,1)\}$, then $c\left(1\right)=c\left(\omega\right)=c\left(\omega^2\right)=0$ yields that $1=0$ or $a_{i}=0$ for some $1\leq i\leq 5$, a contradiction.
\end{itemize}

Thirdly, we suppose that $c(x)$ is a codeword in $\mathcal{C}$ with Hamming weight $7$ and symbol-pair weight $10$.
There are five subcases to be discussed:
\begin{itemize}
  \item
  For the subcase of $c(x)=1+a_{1}\,x+a_{2}\,x^{2}+a_{3}\,x^{3}+a_{4}\,x^{4}+a_{5}\,x^{l_{1}}+a_{6}\,x^{l_{2}}$ with $6\leq l_{1}<l_{2}\leq 3p-2$ and $a_{i}\in \mathbb{F}_p^{*}$ for any $1\leq i\leq6$.
  \begin{itemize}
    \item
    If $\mathcal{CS}\left(l_{1},\,l_{2}\right)\in\left\{(0,\,0),\,(0,\,1),\,(1,\,1)\right\}$, then by $c\left(1\right)=c\left(\omega\right)=c\left(\omega^2\right)=0$, one can immediately deduce $a_{2}=0$, a contradiction.

    \item
    If $\mathcal{CS}\left(l_{1},\,l_{2}\right)\in\left\{(1,\,2),\,(2,\,2)\right\}$, then $c^{\left(1\right)}\left(1\right)=c^{\left(1\right)}\left(\omega\right)=c^{\left(1\right)}\left(\omega^2\right)=0$ yields that $a_{3}=0$, a contradiction.

    \item
    If $\mathcal{CS}\left(l_{1},\,l_{2}\right)=(0,\,2)$, then by $c\left(1\right)=c\left(\omega\right)=c\left(\omega^2\right)=c^{\left(1\right)}\left(1\right)=c^{\left(1\right)}\left(\omega\right)=c^{\left(1\right)}\left(\omega^2\right)=0$, one has $a_{4}=0$, a contradiction.
  \end{itemize}

  \item
  For the subcase of $c(x)=1+a_{1}\,x+a_{2}\,x^{2}+a_{3}\,x^{3}+a_{4}\,x^{l_{1}}+a_{5}\,x^{l_{1}+1}+a_{6}\,x^{l_{2}}$ with $6\leq (l_{1}+1)<l_{2}\leq 3p-2$ and $a_{i}\in \mathbb{F}_p^{*}$ for any $1\leq i\leq6$.
  \begin{itemize}
    \item
    If $\mathcal{V}\left(l_{1},\,l_{2}\right)\in\left\{(0,\,0),\,(0,\,1),\,(2,\,0),\,(2,\,2)\right\}$, then $c\left(1\right)=c\left(\omega\right)=c\left(\omega^2\right)=0$ implies that $a_{1}=0$ or $a_{2}=0$, a contradiction.

    \item
    If $\mathcal{V}\left(l_{1},\,l_{2}\right)\in\left\{(0,\,2),\,(1,\,0),\,(1,\,1),\,(1,\,2),\,(2,\,1)\right\}$, then by $c\left(1\right)=c\left(\omega\right)=c\left(\omega^2\right)=c^{\left(1\right)}\left(1\right)=c^{\left(1\right)}\left(\omega\right)=c^{\left(1\right)}\left(\omega^2\right)=0$, one can immediately get $a_{i}=0$ for some $4\leq i\leq6$, a contradiction.
  \end{itemize}

  \item
  For the subcase of $c(x)=1+a_{1}\,x+a_{2}\,x^{2}+a_{3}\,x^{l_{1}}+a_{4}\,x^{l_{1}+1}+a_{5}\,x^{l_{1}+2}+a_{6}\,x^{l_{2}}$ with $6\leq (l_{1}+2)<l_{2}\leq 3p-2$ and $a_{i}\in \mathbb{F}_p^{*}$ for any $1\leq i\leq6$.
  For any $6\leq (l_{1}+2)<l_{2}\leq 3p-2$, it follows from $c\left(1\right)=c\left(\omega\right)=c\left(\omega^2\right)=c^{\left(1\right)}\left(1\right)
  =c^{\left(1\right)}\left(\omega\right)=c^{\left(1\right)}\left(\omega^2\right)=0$ that $a_{4}=0$ or $a_{5}=0$, a contradiction.

  \item
  For the subcase of $c(x)=1+a_{1}\,x+a_{2}\,x^{2}+a_{3}\,x^{l_{1}}+a_{4}\,x^{l_{1}+1}+a_{5}\,x^{l_{2}}+a_{6}\,x^{l_{2}+1}$ with $5\leq (l_{1}+1)<l_{2}\leq 3p-3$ and $a_{i}\in \mathbb{F}_p^{*}$ for any $1\leq i\leq6$.
  \begin{itemize}
    \item
    If $\mathcal{CS}\left(l_{1},\,l_{2}\right)\in\left\{(0,\,0),\,(1,\,1),\,(2,\,2)\right\}$, then $c\left(1\right)=c\left(\omega\right)=c\left(\omega^2\right)=0$ indicates that $1=0$ or $a_{i}=0$ for some $1\leq i\leq2$, a contradiction.

    \item
    If $\mathcal{CS}\left(l_{1},\,l_{2}\right)\in\left\{(0,\,1),\,(0,\,2),\,(1,\,2)\right\}$, then by $c\left(1\right)=c\left(\omega\right)=c\left(\omega^2\right)=c^{\left(1\right)}\left(1\right)=c^{\left(1\right)}\left(\omega\right)=c^{\left(1\right)}\left(\omega^2\right)=0$, one can conclude that $a_{5}=0$ or $a_{6}=0$, a contradiction.
  \end{itemize}
\end{itemize}

Fourthly, we assume that $c(x)$ is a codeword in $\mathcal{C}$ with Hamming weight $7$ and symbol-pair weight $11$.
There are three subcases to be considered:
\begin{itemize}
  \item
  For the subcase of $c(x)=1+a_{1}\,x+a_{2}\,x^{2}+a_{3}\,x^{3}+a_{4}\,x^{l_{1}}+a_{5}\,x^{l_{2}}+a_{6}\,x^{l_{3}}$ with $5\leq l_{1}<l_{2}<l_{3}\leq 3p-2$ and $a_{i}\in \mathbb{F}_p^{*}$ for any $1\leq i\leq6$.
  \begin{itemize}
    \item
    If $\mathcal{CS}\left(l_{1},\,l_{2},\,l_{3}\right)\in\left\{(0,\,1,\,2),\,
    (1,\,1,\,2),\,(1,\,2,\,2)\right\}$, then by
    $c\left(1\right)=c\left(\omega\right)=c\left(\omega^2\right)=c^{\left(1\right)}\left(1\right)=c^{\left(1\right)}\left(\omega\right)=c^{\left(1\right)}\left(\omega^2\right)=0$, one can immediately have $a_{i}=0$ for some $1\leq i\leq6$, a contradiction.

    \item
    If $\mathcal{CS}\left(l_{1},\,l_{2},\,l_{3}\right)\notin\left\{(0,\,1,\,2),\,(1,\,1,\,2),\,(1,\,2,\,2)\right\}$, then $c\left(1\right)=c\left(\omega\right)=c\left(\omega^2\right)=0$ implies that $a_{1}=0$ or $a_{2}=0$, a contradiction.
  \end{itemize}

  \item
  For the subcase of $c(x)=1+a_{1}\,x+a_{2}\,x^{2}+a_{3}\,x^{l_{1}}+a_{4}\,x^{l_{1}+1}+a_{5}\,x^{l_{2}}+a_{6}\,x^{l_{3}}$ with $5\leq (l_{1}+1)<l_{2}<l_{3}\leq 3p-2$ and $a_{i}\in \mathbb{F}_p^{*}$ for any $1\leq i\leq6$.
  \begin{itemize}
    \item
    If $\mathcal{V}\left(l_{1},\,l_{2},\,l_{3}\right)\in\left\{(1,\,1,\,1),\,(1,\,1,\,2)
    ,\,(1,\,2,\,1),\,(1,\,2,\,2)\right\}$, then by $c\left(1\right)=c\left(\omega\right)=c\left(\omega^2\right)=0$, one can derive that $1=0$, a contradiction.

    \item
    If $\mathcal{V}\left(l_{1},\,l_{2},\,l_{3}\right)\notin\left\{(1,\,1,\,1),\,(1,\,1,\,2)
    ,\,(1,\,2,\,1),\,(1,\,2,\,2)\right\}$, then $c\left(1\right)=c\left(\omega\right)=c\left(\omega^2\right)=c^{\left(1\right)}\left(1\right)=c^{\left(1\right)}\left(\omega\right)=c^{\left(1\right)}\left(\omega^2\right)=0$ induces that $a_{i}=0$ for some $1\leq i\leq6$, a contradiction.
  \end{itemize}

  \item
  For the subcase of $c(x)=1+a_{1}\,x+a_{2}\,x^{l_{1}}+a_{3}\,x^{l_{1}+1}+a_{4}\,x^{l_{2}}+a_{5}\,x^{l_{2}+1}+a_{6}\,x^{l_{3}}$ with $4\leq (l_{1}+1)<l_{2}<l_{3}\leq 3p-2$ and $a_{i}\in \mathbb{F}_p^{*}$ for any $1\leq i\leq6$.
  \begin{itemize}
    \item
    If $\mathcal{V}\left(l_{1},\,l_{2},\,l_{3}\right)\in\left\{(1,\,1,\,1),\,(1,\,1,\,2)\right\}$, then $c\left(1\right)=c\left(\omega\right)=c\left(\omega^2\right)=0$ yields $1=0$, a contradiction.

    \item
    If $\mathcal{V}\left(l_{1},\,l_{2},\,l_{3}\right)\in\left\{(0,\,0,\,0),\,(0,\,0,\,1)\right\}$, then by $c\left(1\right)=c\left(\omega\right)=c\left(\omega^2\right)
    =c^{\left(1\right)}\left(1\right)=c^{\left(1\right)}\left(\omega\right)
    =c^{\left(1\right)}\left(\omega^2\right)=c^{\left(2\right)}\left(1\right)=c^{\left(2\right)}\left(\omega\right)=0$, one can deduce that $a_{4}=0$ or $a_{5}=0$, a contradiction.

    \item
    If $\mathcal{V}\left(l_{1},\,l_{2},\,l_{3}\right)\notin\left\{(0,\,0,\,0),\,(0,\,0,\,1),\,(1,\,1,\,1),\,(1,\,1,\,2)\right\}$, then $c\left(1\right)=c\left(\omega\right)=c\left(\omega^2\right)=c^{\left(1\right)}\left(1\right)=c^{\left(1\right)}\left(\omega\right)=c^{\left(1\right)}\left(\omega^2\right)=0$ leads to $1=0$ or $a_{i}=0$ for some $1\leq i\leq6$, a contradiction.
  \end{itemize}
\end{itemize}

Fifthly, we suppose that $c(x)$ is a codeword in $\mathcal{C}$ with Hamming weight $8$ and symbol-pair weight $10$.
There are four subcases to be considered:
\begin{itemize}
  \item
  For the subcase of $c(x)=1+a_{1}\,x+a_{2}\,x^{2}+a_{3}\,x^{3}+a_{4}\,x^{4}+a_{5}\,x^{5}+a_{6}\,x^{6}+a_{7}\,x^{l}$ with $8\leq l\leq 3p-2$ and $a_{i}\in \mathbb{F}_p^{*}$ for any $1\leq i\leq7$.
  For any $8\leq l\leq 3p-2$, it can be verified that $c\left(1\right)=c\left(\omega\right)=c\left(\omega^2\right)=c^{\left(1\right)}\left(1\right)
  =c^{\left(1\right)}\left(\omega\right)=c^{\left(1\right)}\left(\omega^2\right)=0$ induces that $a_{4}=0$ or $a_{5}=0$, a contradiction.

  \item
  For the subcase of $c(x)=1+a_{1}\,x+a_{2}\,x^{2}+a_{3}\,x^{3}+a_{4}\,x^{4}+a_{5}\,x^{5}+a_{6}\,x^{l}+a_{7}\,x^{l+1}$ with $7\leq l\leq 3p-3$ and $a_{i}\in \mathbb{F}_p^{*}$ for any $1\leq i\leq7$.
  For any $7\leq l\leq 3p-3$, $c\left(1\right)=c\left(\omega\right)=c\left(\omega^2\right)=c^{\left(1\right)}\left(1\right)
  =c^{\left(1\right)}\left(\omega\right)=c^{\left(1\right)}\left(\omega^2\right)=0$ implies that $a_{i}=0$ for some $3\leq i\leq5$, a contradiction.

  \item
  For the subcase of $c(x)=1+a_{1}\,x+a_{2}\,x^{2}+a_{3}\,x^{3}+a_{4}\,x^{4}+a_{5}\,x^{l}+a_{6}\,x^{l+1}+a_{7}\,x^{l+2}$ with $6\leq l\leq 3p-4$ and $a_{i}\in \mathbb{F}_p^{*}$ for any $1\leq i\leq7$.
  For any $6\leq l\leq 3p-4$, it follows from $c\left(1\right)=c\left(\omega\right)=c\left(\omega^2\right)=c^{\left(1\right)}\left(1\right)
  =c^{\left(1\right)}\left(\omega\right)=c^{\left(1\right)}\left(\omega^2\right)=0$ that $a_{i}=0$ for some $5\leq i\leq7$, a contradiction.

  \item
  For the subcase of $c(x)=1+a_{1}\,x+a_{2}\,x^{2}+a_{3}\,x^{3}+a_{4}\,x^{l}+a_{5}\,x^{l+1}+a_{6}\,x^{l+2}+a_{7}\,x^{l+3}$ with $5\leq l\leq 3p-5$ and $a_{i}\in \mathbb{F}_p^{*}$ for any $1\leq i\leq7$.
  For any $5\leq l\leq 3p-5$, by $c\left(1\right)=c\left(\omega\right)=c\left(\omega^2\right)=c^{\left(1\right)}\left(1\right)
  =c^{\left(1\right)}\left(\omega\right)=c^{\left(1\right)}\left(\omega^2\right)=0$, one can derive that $a_{5}=0$ or $a_{6}=0$, a contradiction.
\end{itemize}

Sixthly, we suppose that $c(x)$ is a codeword in $\mathcal{C}$ with Hamming weight $8$ and symbol-pair weight $11$.
Without loss of generality, it suffices to consider the following five subcases:
\begin{itemize}
  \item
  For the subcase of $c(x)=1+a_{1}\,x+a_{2}\,x^{2}+a_{3}\,x^{3}+a_{4}\,x^{4}+a_{5}\,x^{5}+a_{6}\,x^{l_{1}}+a_{7}\,x^{l_{2}}$ with $7\leq l_{1}<l_{2}\leq 3p-2$ and $a_{i}\in \mathbb{F}_p^{*}$ for any $1\leq i\leq7$.
  \begin{itemize}
    \item
    If $\mathcal{CS}\left(l_{1},\,l_{2}\right)=(1,\,2)$, then it can be verified that by $c^{\left(1\right)}\left(1\right)=c^{\left(1\right)}\left(\omega\right)=c^{\left(1\right)}\left(\omega^2\right)=0$, one can obtain $a_{3}=0$, a contradiction.

    \item
    If $\mathcal{CS}\left(l_{1},\,l_{2}\right)\ne (1,\,2)$, then $c\left(1\right)=c\left(\omega\right)=c\left(\omega^2\right)=c^{\left(1\right)}\left(1\right)
    =c^{\left(1\right)}\left(\omega\right)=c^{\left(1\right)}\left(\omega^2\right)=0$ indicates $a_{4}=0$ or $a_{5}=0$, a contradiction.
  \end{itemize}

  \item
  For the subcase of $c(x)=1+a_{1}\,x+a_{2}\,x^{2}+a_{3}\,x^{3}+a_{4}\,x^{4}+a_{5}\,x^{l_{1}}+a_{6}\,x^{l_{1}+1}+a_{7}\,x^{l_{2}}$ with $7\leq (l_{1}+1)<l_{2}\leq 3p-2$ and $a_{i}\in \mathbb{F}_p^{*}$ for any $1\leq i\leq7$.
  \begin{itemize}
    \item
    If $\mathcal{V}\left(l_{1},\,l_{2}\right)=(1,\,2)$, then by $c^{\left(1\right)}\left(1\right)=c^{\left(1\right)}\left(\omega\right)=c^{\left(1\right)}\left(\omega^2\right)=0$, one can get $a_{3}=0$, a contradiction.

    \item
    If $\mathcal{V}\left(l_{1},\,l_{2}\right)\ne(1,\,2)$, then $c\left(1\right)=c\left(\omega\right)=c\left(\omega^2\right)=c^{\left(1\right)}\left(1\right)
    =c^{\left(1\right)}\left(\omega\right)=c^{\left(1\right)}\left(\omega^2\right)=0$ implies $a_{i}=0$ for some $1\leq i\leq7$, a contradiction.
  \end{itemize}

  \item
  For the subcase of $c(x)=1+a_{1}\,x+a_{2}\,x^{2}+a_{3}\,x^{3}+a_{4}\,x^{l_{1}}+a_{5}\,x^{l_{1}+1}+a_{6}\,x^{l_{1}+2}+a_{7}\,x^{l_{2}}$ with $7\leq (l_{1}+2)<l_{2}\leq 3p-2$ and $a_{i}\in \mathbb{F}_p^{*}$ for any $1\leq i\leq7$.
  For any $7\leq (l_{1}+2)<l_{2}\leq 3p-2$, it can be verified that by $c\left(1\right)=c\left(\omega\right)=c\left(\omega^2\right)=c^{\left(1\right)}\left(1\right)
  =c^{\left(1\right)}\left(\omega\right)=c^{\left(1\right)}\left(\omega^2\right)=0$, one can deduce $a_{i}=0$ for some $4\leq i\leq 6$, a contradiction.

  \item
  For the subcase of $c(x)=1+a_{1}\,x+a_{2}\,x^{2}+a_{3}\,x^{3}+a_{4}\,x^{l_{1}}+a_{5}\,x^{l_{1}+1}+a_{6}\,x^{l_{2}}+a_{7}\,x^{l_{2}+1}$ with $6\leq (l_{1}+1)<l_{2}\leq 3p-3$ and $a_{i}\in \mathbb{F}_p^{*}$ for any $1\leq i\leq7$.
  \begin{itemize}
    \item
    If $\mathcal{V}\left(l_{1},\,l_{2}\right)=(1,\,1)$, then  $c^{\left(1\right)}\left(1\right)=c^{\left(1\right)}\left(\omega\right)=c^{\left(1\right)}\left(\omega^2\right)=0$ yields $a_{3}=0$, a contradiction.

    \item
    If $\mathcal{V}\left(l_{1},\,l_{2}\right)\ne(1,\,1)$, then by $c\left(1\right)=c\left(\omega\right)=c\left(\omega^2\right)=c^{\left(1\right)}\left(1\right)
    =c^{\left(1\right)}\left(\omega\right)=c^{\left(1\right)}\left(\omega^2\right)=0$, one can derive $a_{i}=0$ for some $1\leq i\leq 7$, a contradiction.
  \end{itemize}

  \item
  For the subcase of $c(x)=1+a_{1}\,x+a_{2}\,x^{2}+a_{3}\,x^{l_{1}}+a_{4}\,x^{l_{1}+1}+a_{5}\,x^{l_{1}+2}+a_{6}\,x^{l_{2}}+a_{7}\,x^{l_{2}+1}$ with $6\leq (l_{1}+2)<l_{2}\leq 3p-3$ and $a_{i}\in \mathbb{F}_p^{*}$ for any $1\leq i\leq7$.
  For any $6\leq (l_{1}+2)<l_{2}\leq 3p-3$, it follows from $c\left(1\right)=c\left(\omega\right)=c\left(\omega^2\right)=c^{\left(1\right)}\left(1\right)=c^{\left(1\right)}\left(\omega\right)=c^{\left(1\right)}\left(\omega^2\right)=0$ that $a_{i}=0$ for some $1\leq i\leq 5$, a contradiction.
\end{itemize}

Finally, we assume that $c(x)$ is a codeword in $\mathcal{C}$ with Hamming weight $9$ and symbol-pair weight $11$.
There are four subcases to be discussed:
\begin{itemize}
  \item
  For the subcase of $c(x)=1+a_{1}\,x+a_{2}\,x^{2}+a_{3}\,x^{3}+a_{4}\,x^{4}+a_{5}\,x^{5}+a_{6}\,x^{6}+a_{7}\,x^{7}+a_{8}\,x^{l}$ with $9\leq l\leq 3p-2$ and $a_{i}\in \mathbb{F}_p^{*}$ for any $1\leq i\leq8$.
  \begin{itemize}
    \item
    If $\mathcal{V}(l)\in\{0,\,1\}$, then by $c\left(1\right)=c\left(\omega\right)=c\left(\omega^2\right)=c^{\left(1\right)}\left(1\right)
    =c^{\left(1\right)}\left(\omega\right)=c^{\left(1\right)}\left(\omega^2\right)=0$, one can obtain $a_{5}=0$, a contradiction.

    \item
    If $\mathcal{V}(l)=2$, then $c\left(1\right)=c\left(\omega\right)
    =c\left(\omega^2\right)=c^{\left(1\right)}\left(1\right)
    =c^{\left(1\right)}\left(\omega\right)=c^{\left(1\right)}\left(\omega^2\right)=0$ induces
    \begin{equation*}
    \left\{
    \begin{array}{l}
      1+a_{3}+a_{6}=a_{1}+a_{4}+a_{7}=a_{2}+a_{5}+a_{8}=0,\\[2mm]
      3\,a_{3}+6\,a_{6}=a_{1}+4\,a_{4}+7\,a_{7}=2\,a_{2}+5\,a_{5}+l\,a_{8}=0
    \end{array}
    \right.
    \end{equation*}
    which indicates that
    \begin{equation}\label{eq3p12-9-11-8101}
      a_{2}=\frac{l-5}{3}\,a_{8},\quad a_{3}=-2,\quad a_{4}=-2\,a_{7},\quad a_{5}=-\frac{l-2}{3}\,a_{8},\quad a_{6}=1.
    \end{equation}
    The fact $c^{\left(2\right)}\left(1\right)=c^{\left(2\right)}\left(\omega\right)=0$ indicates
    \begin{equation*}
    \left\{
    \begin{array}{l}
      2\,a_{2}+20\,a_{5}+l\left(l-1\right)a_{8}+6\,a_{3}+30\,a_{6}+12\,a_{4}+42\,a_{7}=0,\\[2mm]
      2\,a_{2}+20\,a_{5}+l\left(l-1\right)a_{8}+\left(6\,a_{3}+30\,a_{6}\right)\omega+\left(12\,a_{4}+42\,a_{7}\right)\omega^2=0\\
    \end{array}
    \right.
    \end{equation*}
    which yields
    \begin{equation*}
      a_{7}=\omega \quad {\rm and} \quad a_{8}=\frac{18\,\omega^2}{\left(l-2\right)\left(l-5\right)}
    \end{equation*}
    due to (\ref{eq3p12-9-11-8101}).
    It follows from $c^{\left(3\right)}\left(1\right)=0$ that
    \begin{equation*}
      6\,a_{3}+24\,a_{4}+60\,a_{5}+120\,a_{6}+210\,a_{7}+l\left(l-1\right)\left(l-2\right)a_{8}=0
    \end{equation*}
    which yields $l=2-3\,\omega^2$.
    By $c^{\left(4\right)}\left(1\right)=0$, one can get
    \begin{equation*}
      24\,a_{4}+120\,a_{5}+360\,a_{6}+840\,a_{7}
      +l\left(l-1\right)\left(l-2\right)\left(l-3\right)\,a_{8}=0.
    \end{equation*}
    This implies $\left(l^2+l-12\right)\omega+24=0$, which leads to
    \begin{equation*}
      \left(\left(2-3\,\omega^2\right)^2+2-3\,\omega^2-12\right)\omega+24=0.
    \end{equation*}
    Thus $-5\,\omega=0$, a contradiction.
  \end{itemize}

  \item
  For the subcase of $c(x)=1+a_{1}\,x+a_{2}\,x^{2}+a_{3}\,x^{3}+a_{4}\,x^{4}+a_{5}\,x^{5}+a_{6}\,x^{6}+a_{7}\,x^{l}+a_{8}\,x^{l+1}$ with $8\leq l\leq 3p-3$ and $a_{i}\in \mathbb{F}_p^{*}$ for any $1\leq i\leq8$.
  \begin{itemize}
    \item
    If $\mathcal{V}(l)\in\{0,\,2\}$, then by $c\left(1\right)=c\left(\omega\right)=c\left(\omega^2\right)=c^{\left(1\right)}\left(1\right)
    =c^{\left(1\right)}\left(\omega\right)=c^{\left(1\right)}\left(\omega^2\right)=0$, one can conclude $a_{4}=0$ or $a_{5}=0$, a contradiction.

    \item
    If $\mathcal{V}(l)=1$, then $c\left(1\right)=c\left(\omega\right)=c\left(\omega^2\right)=c^{\left(1\right)}\left(1\right)
    =c^{\left(1\right)}\left(\omega\right)=c^{\left(1\right)}\left(\omega^2\right)=0$ implies
    \begin{equation}\label{eq3p12-9-11-7201}
      a_{2}=\frac{l-4}{3}\,a_{8},\quad a_{3}=-2,\quad a_{4}=-\frac{l-1}{3}\,a_{7},\quad a_{5}=-\frac{l-1}{3}\,a_{8},\quad a_{6}=1.
    \end{equation}
    It can be verified by (\ref{eq3p12-9-11-7201}) and  $c^{\left(2\right)}\left(1\right)=c^{\left(2\right)}\left(\omega\right)=0$ that
    \begin{equation*}
      a_{7}=\frac{18\,\omega}{\left(l-1\right)\left(l-4\right)} \quad {\rm and} \quad a_{8}=\frac{18\,\omega^2}{\left(l-1\right)\left(l-4\right)}.
    \end{equation*}
    According to $c^{\left(3\right)}\left(1\right)=0$, one can obtain $l=-\omega^2-4\,\omega$.
    Hence $c^{\left(4\right)}\left(1\right)=0$ yields $3\,\omega^2=0$, a contradiction.
  \end{itemize}

  \item
  For the subcase of $c(x)=1+a_{1}\,x+a_{2}\,x^{2}+a_{3}\,x^{3}+a_{4}\,x^{4}+a_{5}\,x^{5}+a_{6}\,x^{l}+a_{7}\,x^{l+1}+a_{8}\,x^{l+2}$ with $7\leq l\leq 3p-4$ and $a_{i}\in \mathbb{F}_p^{*}$ for any $1\leq i\leq8$.
  \begin{itemize}
    \item
    If $\mathcal{V}(l)=0$, then by $c\left(1\right)=c\left(\omega\right)=c\left(\omega^2\right)=c^{\left(1\right)}\left(1\right)
    =c^{\left(1\right)}\left(\omega\right)=c^{\left(1\right)}\left(\omega^2\right)=0$, one can derive
    \begin{equation}\label{eq3p12-9-11-6301}
      a_{2}=\frac{l-3}{3}\,a_{8},\quad a_{3}=-\frac{l}{l-3},\quad a_{4}=-\frac{l}{3}\,a_{7},\quad a_{5}=-\frac{l}{3}\,a_{8},\quad a_{6}=\frac{3}{l-3}.
    \end{equation}
    By (\ref{eq3p12-9-11-6301}) and  $c^{\left(2\right)}\left(1\right)=c^{\left(2\right)}\left(\omega\right)=0$, one has
    \begin{equation*}
      a_{7}=\frac{3\,\omega}{l-3} \quad {\rm and} \quad a_{8}=\frac{3\,\omega^2}{l-3}.
    \end{equation*}
    Then $c^{\left(3\right)}\left(1\right)=0$ implies
    \begin{equation*}
      6\,a_{3}+24\,a_{4}+60\,a_{5}+l\left(l-1\right)\left(l-2\right)a_{6}
      +l\left(l-1\right)\left(l+1\right)a_{7}
      +l\left(l+1\right)\left(l+2\right)a_{8}=0
    \end{equation*}
    which yields $\omega^2=\omega$, a contradiction.

    \item
    If $\mathcal{V}(l)=1$, then $c\left(1\right)=c\left(\omega\right)=c\left(\omega^2\right)=c^{\left(1\right)}\left(1\right)
    =c^{\left(1\right)}\left(\omega\right)=c^{\left(1\right)}\left(\omega^2\right)=0$ indicates
    \begin{equation}\label{eq3p12-9-11-6302}
      a_{2}=\frac{l-4}{3}\,a_{7},\quad a_{3}=-\frac{l+2}{l-1},\quad a_{4}=-\frac{l-1}{3}\,a_{6},\quad a_{5}=-\frac{l-1}{3}\,a_{7},\quad a_{8}=\frac{3}{l-1}.
    \end{equation}
    It follows from (\ref{eq3p12-9-11-6302}) and  $c^{\left(2\right)}\left(1\right)=c^{\left(2\right)}\left(\omega\right)=0$ that
    \begin{equation*}
      a_{6}=\frac{3\left(l+2\right)\omega}{\left(l-1\right)\left(l-4\right)} \quad {\rm and} \quad a_{7}=\frac{3\left(l+2\right)\omega^2}{\left(l-1\right)\left(l-4\right)}.
    \end{equation*}
    By $c^{\left(3\right)}\left(1\right)=0$, one can deduce that $3\,\omega^2=0$, a contradiction.

    \item
    If $\mathcal{V}(l)=2$, then by $c\left(1\right)=c\left(\omega\right)=c\left(\omega^2\right)=c^{\left(1\right)}\left(1\right)
    =c^{\left(1\right)}\left(\omega\right)=c^{\left(1\right)}\left(\omega^2\right)=0$, one has
    \begin{equation}\label{eq3p12-9-11-6303}
      a_{2}=\frac{l-5}{3}\,a_{6},\, a_{3}=-\frac{l+1}{3}\,a_{7},\, a_{4}=-\frac{l+1}{3}\,a_{8},\, a_{5}=-\frac{l-2}{3}\,a_{6},\, a_{7}=\frac{3}{l-2}.
    \end{equation}
    Due to (\ref{eq3p12-9-11-6303}) and  $c^{\left(2\right)}\left(1\right)=c^{\left(2\right)}\left(\omega\right)=0$, one can immediately obtain that
    \begin{equation*}
      a_{6}=\frac{3\left(l+1\right)\omega^2}{\left(l-2\right)\left(l-5\right)} \quad {\rm and} \quad a_{8}=\frac{3\,\omega}{l-2}.
    \end{equation*}
    Thus $c^{\left(3\right)}\left(1\right)=0$ leads to $3=0$, a contradiction.
  \end{itemize}

  \item
  For the subcase of $c(x)=1+a_{1}\,x+a_{2}\,x^{2}+a_{3}\,x^{3}+a_{4}\,x^{4}+a_{5}\,x^{l}+a_{6}\,x^{l+1}+a_{7}\,x^{l+2}+a_{8}\,x^{l+3}$ with $6\leq l\leq 3p-5$ and $a_{i}\in \mathbb{F}_p^{*}$ for any $1\leq i\leq8$.
  \begin{itemize}
    \item
    If $\mathcal{V}(l)\in\{0,\,1\}$, then by $c\left(1\right)=c\left(\omega\right)=c\left(\omega^2\right)=c^{\left(1\right)}\left(1\right)
    =c^{\left(1\right)}\left(\omega\right)=c^{\left(1\right)}\left(\omega^2\right)=0$, one gets $a_{6}=0$ or $a_{7}=0$, a contradiction.

    \item
    If $\mathcal{V}(l)=2$, then $c\left(1\right)=c\left(\omega\right)=c\left(\omega^2\right)=c^{\left(1\right)}\left(1\right)
    =c^{\left(1\right)}\left(\omega\right)=c^{\left(1\right)}\left(\omega^2\right)=0$ induces
    \begin{equation}\label{eq3p12-9-11-5401}
      a_{2}=\frac{3}{l-2}\,a_{8},\, a_{3}=-\frac{l+1}{3}\,a_{6},\, a_{4}=-\frac{l+1}{3}\,a_{7},\, a_{5}=-\frac{l+1}{l-2}\,a_{8},\, a_{6}=\frac{3}{l-2}.
    \end{equation}
    By (\ref{eq3p12-9-11-5401}) and  $c^{\left(2\right)}\left(1\right)=c^{\left(2\right)}\left(\omega\right)=0$, one can derive
    \begin{equation*}
      a_{7}=\frac{3\,\omega}{l-2} \quad {\rm and} \quad a_{8}=\omega^2.
    \end{equation*}
    It follows from $c^{\left(3\right)}\left(1\right)=0$ that $l=-3\,\omega^2-1$.
    Then $c^{\left(4\right)}\left(1\right)=0$ yields $l\,\omega+2=-\omega-1=0$, a contradiction.
  \end{itemize}
\end{itemize}

As a result, $\mathcal{C}$ is an MDS $\left(3p,\,12\right)_{p}$ symbol-pair code.
The desired result follows.
\end{proof}

\begin{remark}
Note that if $\mathcal{C}$ is a repeated-root cyclic code of length $3p$ over $\mathbb{F}_p$ with generator polynomial
\begin{equation*}
  g(x)=\left(x-1\right)^{4}\left(x-\omega\right)^{3}\left(x-\omega^{2}\right)^2
\end{equation*}
where $\omega$ is a primitive third root of unity in $\mathbb{F}_{p}$.
Due to Theorem \ref{thmMDS3p10p}, we can conclude that $\mathcal{C}$ is an AMDS $\left(3p,\,10\right)_{p}$ symbol-pair code.
Indeed, by Lemma \ref{lemdistance}, one can immediately get $d_{H}(\mathcal{C})=5$.
Since $\left(x^3-1\right)\,\big|\,g(x)$ and $2<5=d_{H}(\mathcal{C})<3p-(3p-9)=9$, Lemma \ref{lemdh3} indicates that $d_{p}(\mathcal{C})\geq 8$.
It follows from the proof of Theorem \ref{thmMDS3p10p} that there does not exist a codeword $c(x)$ in $\mathcal{C}$ with $(w_{H}(c(x)),w_{p}(c(x)))=(5,\,8)$, $(5,\,9)$, $(6,\,8)$, $(6,\,9)$, $(7,\,8)$, $(7,\,9)$ or $(8,\,9)$.
Therefore, the inequality (\ref{eqdhdp}) shows that $\mathcal{C}$ is an AMDS $\left(3p,\,10\right)_{p}$ symbol-pair code.
\end{remark}

In what follows, we present two examples to illustrate the result in Theorems \ref{thmMDS3p10p} and \ref{thmMDS3p12p}.

\begin{example}
(1) Let $\mathcal{C}$ be a repeated-root cyclic code of length $21$ over $\mathbb{F}_{7}$ with generator polynomial
\begin{equation*}
  g(x)=\left(x-1\right)^{4}\left(x-2\right)^2\left(x-2^2\right)^2.
\end{equation*}
By the computation software MAGMA, it can be verified that $\mathcal{C}$ is a $[21,\,13,\,5]$ code and the minimum symbol-pair distance of $\mathcal{C}$ is $10$, which coincides with our result in Theorem \ref{thmMDS3p10p}.

(2) Let $\mathcal{C}$ be a repeated-root cyclic code of length $21$ over $\mathbb{F}_{7}$ with generator polynomial
\begin{equation*}
  g(x)=\left(x-1\right)^{5}\left(x-2\right)^{3}\left(x-2^2\right)^{2}.
\end{equation*}
MAGMA experiments yield that $\mathcal{C}$ is a $[21,\,11,\,6]$ code and the minimum symbol-pair distance of $\mathcal{C}$ is $12$, which is consistent with our result in Theorem \ref{thmMDS3p12p}.
\end{example}

\section{Conclusions}

In this paper, for $n=3p$, we construct two new classes of MDS symbol-pair codes over $\mathbb{F}_{p}$ with $p$ an odd prime by employing repeated-root cyclic codes:
\begin{itemize}
  \item $[3p,\,3p-8,\,5]$ code with $d_{p}=10$;
  \item $[3p,\,3p-10,\,6]$ code with $d_{p}=12$.
\end{itemize}
As mentioned in Table \ref{Tab-constacyclic-MDS}, these codes poss minimum symbol-pair distance bigger than all the known MDS symbol-pair codes from constacyclic codes.
Note that alongside with larger minimum symbol-pair distance, much more cases need to be considered, which has not been explored.

\end{document}